\documentclass[a4paper,11pt]{amsart}

\usepackage[hang,flushmargin]{footmisc}

\usepackage[utf8]{inputenc}

\usepackage{amsmath,amssymb,amsthm,a4wide}

\usepackage{mathtools}
\mathtoolsset{showonlyrefs}

\usepackage{hyperref}

\usepackage{graphicx,tikz, caption}

\usepackage{framed}

\usepackage{orcidlink}
\usepackage{framed}

\newtheorem{theorem}{Theorem}[section]
\newtheorem*{thm}{Theorem}
\newtheorem{problem}{Problem}
\newtheorem{proposition}[theorem]{Proposition}

\newtheorem{conjecture}[theorem]{Conjecture}
\newtheorem{lemma}[theorem]{Lemma}

\newtheorem*{problem*}{Problem}

\theoremstyle{definition}
\newtheorem{definition}[theorem]{Definition}

\newtheorem*{remark}{Remark}
\newtheorem{question}{Question}

\newcommand{\supp}{\operatorname{supp}}

\renewcommand{\phi}{\varphi}

\providecommand{\norm}[1]{\left\lVert #1 \right\rVert}

\newcommand{\R}{\mathbb{R}}
\newcommand{\C}{\mathbb{C}}

\newcommand{\Rd}{\mathbb{R}^d}
\newcommand{\Z}{\mathbb{Z}}

\newcommand{\F}{\mathcal{F}}
\newcommand{\G}{\mathcal{G}}

\newcommand{\glam}{g_{\lambda}}
\newcommand{\den}{\delta}
\newcommand{\eps}{\varepsilon}

\renewcommand{\l}{\lambda}
\renewcommand{\L}{\Lambda}

\newcommand{\Lt}[1][d]{L^2(\R^{#1})}

\newcommand{\vol}{\textnormal{vol}}

\usepackage{xcolor}

\begin{document}

\title{Quantum paving: When sphere packings meet Gabor frames}

\author[M.~Faulhuber]{Markus Faulhuber \orcidlink{0000-0002-7576-5724}}
\address{Faculty of Mathematics, University of Vienna\newline Oskar-Morgenstern-Platz 1, 1090 Vienna, Austria}
\email{markus.faulhuber@univie.ac.at}

\author[T.~Strohmer]{Thomas Strohmer\orcidlink{0000-0003-2029-3317}}
\address{Department of Mathematics, University of California, Davis, CA 95616, USA }
\email{strohmer@math.ucdavis.edu}

\makeatletter
\@namedef{subjclassname@2020}{\textup{2020} Mathematics Subject Classification}
\makeatother

\begin{abstract}
We introduce the new problems of quantum packing, quantum covering, and quantum paving. These problems arise naturally when considering an algebra of non-commutative operators that is deeply rooted in quantum physics as well as in Gabor analysis. Quantum packing and quantum covering show similarities with energy minimization and the dual problem of polarization. Quantum paving, in turn, aims to simultaneously optimize both quantum packing and quantum covering. Classical sphere packing and covering hint the optimal configurations for our new problems. We present solutions in certain cases, state several conjectures related to quantum paving and discuss some applications.
\end{abstract}

\maketitle

\section{Introduction}\label{s:intro}
The sphere packing problem is an ancient mathematical problem and has fascinated and attracted researchers for centuries~\cite{Conway}. The goal is to arrange non-overlapping spheres of same size in the most economical way, i.e., such that their interiors use up the largest possible portion of Euclidean space. In dimension three, the problem became (in)famous due to the Kepler conjecture which, with great effort, was proven to be correct by T.~Hales \cite{Hal05} about 20 years ago, and four centuries after its statement.

M.~Viazovska's tremendous breakthrough in solving the sphere packing problem in dimension 8 \cite{Via17} and, in joint work with H.~Cohn, A.~Kumar, S.~Miller and D.~Radchenko in dimension 24 \cite{Coh-Via17} is among the greatest recent developments in mathematics.

The cousin of the sphere packing problem is the sphere covering problem, which has gained less attention~\cite{Conway}. The goal is to entirely cover Euclidean space by equally sized spheres, such that the overlap of their interiors is as small as possible. In a way, the two problems are two sides of the same medal. Nonetheless, they are usually treated separately and tools and solutions may be quite different. 

We introduce the new notions of \textit{quantum packing} and \textit{quantum covering}, which are inspired by viewing the construction of optimal Gaussian Gabor frames through the lens of classical sphere  packing and covering. This is not the first time these two fields meet: Gabor systems were used in the article by Manin and Marcolli~\cite{ManMar23} to construct Cohn-Elkies functions which come arbitrarily close to the linear programming bound for the sphere packing problem derived by Cohn and Elkies \cite{CohElk03}. Here, we go in the other directions and seek to use knowledge on sphere packings and coverings to find candidates which solve our newly introduced problems.
They can also be formulated  by a two-sided inequality, involving non-commuting unitary operators, arising from the Schr\"odinger representation of the Heisenberg group.  

Just as its classical pendant, the quantum packing and covering problem is simple in its statement (though certainly more involved than the classical case), but notoriously difficult to solve. In fact, there are additional issues popping up, which do not arise in the classical questions, e.g., it is non-obvious and perhaps not true that the solutions are invariant under scaling (i.e, changing the density of a lattice).

Lastly, we present the notion of \textit{quantum paving}, which aims (in a sense that we will make precise) to simultaneously optimize both quantum packing and quantum covering. As such it is  related to finding the packing-covering constant \cite{SchVal04}.

The problems we discuss have applications in frame theory, time-frequency analysis, wireless communications, and potentially also in quantum error correction.


\section{Formulation of the problems}
We start with introducing the problems, using mainly notation from the textbooks of G.~Folland~\cite{Fol89} and K.~Gröchening~\cite{Gro01}. The details are then given in subsequent sections.

Let $\varphi \in \Lt$ be the normalized standard Gaussian function, i.e.,
\begin{equation}\label{eq:Gaussian}
    \varphi(t) = 2^{d/4} e^{-\pi |t|^2}, \quad t \in \Rd
    \qquad \text{ so } \quad
    \norm{\varphi}_{L^2}=1.
\end{equation}
Let $z = (x,\omega) \in \Rd \times \Rd \cong \R^{2d}$ and $\tau \in \R$, so $(x,\omega;\tau)$ is an element of the Heisenberg group ~$\mathbb{H}$, and let $\rho(z) = \rho(x,\omega;0)$ denote the unitary operator coming from the Schr\"odinger representation of the element $(x,\omega;0)$ from the Heisenberg group. In particular, up to a phase factor $c \in \C$, $|c|=1$, we have $\rho(z) = c \ \pi(z)$, where $\pi(z)$ is the classical time-frequency shift, i.e., the composition of the translation operator $T_x$ and the modulation operator $M_\omega$, formally defined in~\eqref{eq:Tx_Mw}.

Now, for a lattice $\L$ in phase space $\R^{2d}$, find the optimal constants $A_\L$ and $B_\L$ such that
\begin{equation}\label{eq:frame}
    A_\L \norm{f}_{\Lt}^2 \leq \sum_{\l \in \L} |\langle f, \rho(\l) \varphi \rangle|^2 \leq B_\L \norm{f}_{\Lt}^2, \qquad \forall f \in \Lt.
\end{equation}
The reader familiar with Gabor analysis will identify \eqref{eq:frame} as the frame inequality of the Gabor system $\G\{c \, \pi(\l) \varphi \mid \l \in \L\}$. We can now introduce the problems of quantum packing, covering, and paving.
\begin{framed}
\begin{problem}[Quantum packing]\label{pro:qpacking}
    For any given fixed density, find the configuration $\L_B$ such that $B_{\L_B} \leq B_\L$ for all $\L$ of the same density.
\end{problem}
\medskip
\begin{problem}[Quantum covering]\label{pro:qcovering}
    For any given fixed density, find the configuration $\L_A$ such that $A_{\L_A} \leq A_\L$ for all $\L$ of the same density.
\end{problem}
\medskip
\begin{problem}[Quantum paving]\label{pro:qpaving}
    For any given fixed density, find the configuration $\L_{B/A}$ such that $B_{\L_{B/A}} / A_{\L_{B/A}} \leq B_\L / A_\L$ for all $\L$ of the same density.
\end{problem}
\end{framed}

The ``quantum'' aspect in the aforementioned problems can be traced back along different routes to a non-commuting operator algebra that plays a foundational role in quantum physics. Each of these different routes highlights a different facet of these problems.
In one route, taken in Sections~\ref{s:gabor} and~\ref{s:ofdm} we will encounter {\em coherent states}, i.e., states of minimum uncertainty, which form a fundamental concept in quantum mechanics. Another route will take us to ``paving'' the phase space,  see Section~\ref{s:phase}.
Yet another route revolves around the {\em quantum torus}, discussed briefly in the next subsection.

\subsection{The quantum torus}

For a pair of non-commuting, unitary operators $(U_1$, $U_2)$ on a Hilbert space and $r \in \R/\Z$, the quantum torus~$\mathcal{T}$, also called non-commutative torus \cite{Con94}, is an operator algebra with elements being the formal sums of the form (cf.\ \cite[pp.~217, Ex.~2.$\beta$]{Con94})
\begin{equation}
	\sum_{\l \in \L} c_\l \ U(\l),
	\quad \text{ where } \quad
	U(\l) = U_2(\l_2) U_1(\l_1),
	\quad U_2 U_1 = e^{2 \pi i r} U_1 U_2,
	\quad \text{ and } \quad \l=(\l_1,\l_2) \in \L \subset \R^{2d}.
\end{equation}

In our situation, we use the Hilbert space $\Lt$ of square-integrable functions and our non-commuting operators are the translation operator $T_x$ and the modulation operator $M_\omega$. Their combination leads to time-frequency shifts $\pi(z)$, serving as coordinates for the quantum torus (see \eqref{eq:Tx_Mw}, \eqref{eq:TFshift} and \eqref{eq:comm_relations} below). So, its elements are formally given by
\begin{equation}\label{eq:quantum_torus}
	\sum_{\l \in \L} c_\l \ \pi(\l), \quad \L \subset \R^{2d},
\end{equation}
where $c = (c_\l)_{\l \in \L}$ is a sequence on the lattice $\L$. The properties of the sequence induce an additional structure on the quantum torus $\mathcal{T}_\L$. For example, denoting the bounded operators on $\Lt$ by $\mathcal{B}(\Lt)$, if $c \in \mathcal{S}(\L)$ is a sequence of rapid decay, i.e., it decreases faster than the inverse of any polynomial, then we obtain the \textit{smooth} structure:
\begin{equation}
	\mathcal{T}_\L^\infty = \{T \in \mathcal{B}(\Lt) \mid T = \sum_{\l \in \L} c_\l \ \pi(\l), \, (c_\l) \in \mathcal{S}(\L) \}.
\end{equation}
By assuming the sequence of coefficients to be absolutely summable we obtain:
\begin{equation}
	\mathcal{T}_\L^1 = \{T \in \mathcal{B}(\Lt) \mid T = \sum_{\l \in \L} c_\l \ \pi(\l), \, (c_\l) \in \ell^1(\L) \}.
\end{equation}
Characterizations of these and more structures, also in terms of Gabor frames, may be found in the work of Luef \cite{Lue07}, \cite{Lue09} \cite{Lue11}. For a more general framework, we refer to the article of Manin \cite{Man01} and for our specific setup we refer to \cite[Sec.~5]{LueMan09}.

For us, \eqref{eq:quantum_torus} provides one way of motivating quantum packing and covering, as they concern continuity and invertibility properties of elements of the form \eqref{eq:quantum_torus} which come from a quantum torus. The discipline which studies these properties is \textit{quantum-harmonic analysis} (QHA), introduced by Werner~\cite{Wer84}. QHA aims to formulate classical results from harmonic analysis for operators of the form \eqref{eq:quantum_torus} and has recently also gained new attraction in the context of Gabor analysis \cite{LueSkr21}, \cite{Skr20}, \cite{Skr_PhD}. Depending on the (decay) properties of the sequence $(c_\l)_{\l \in \L}$, we have different kinds of bounded operators forming a quantum torus (see \cite[Sec.~3]{Lue06}), similar to various function spaces on the torus. Finding conditions on the invertibility of the series in \eqref{eq:quantum_torus} is a pendant to Wiener's lemma for Fourier series and falls into the scope of quantum covering. Questions of convergence fall into the scope of quantum packing. The mentioned properties may be studied by means of \textit{quantum thetas}~$\Theta_\L$.
These are operator pendants to classical \textit{theta functions} (see \cite{LueMan09}, \cite{Man01}) and act on functions $f \in \Lt$ by the rule
\begin{equation}\label{eq:quantum_theta}
	\Theta_\L f = \sum_{\l \in \L} \langle \varphi, \pi(\l) \varphi \rangle \pi(\l) f.
\end{equation}
Compared to classical thateta functions, the inner product is the replacement for the Gaussian function, $\pi(\l)$ replaces the complex units and $f$ the point at which we evaluate.
By the \textit{symplectic Poisson summation formula} (see Section~\ref{s:symplectic}) we will see that a quantum theta is just the frame operator in disguise: using the adjoint lattice $\L^\circ$ \eqref{eq:adjoint_lattice}, $\Theta_{\L^\circ}$ is the \textit{Janssen representation} of the frame operator (see \cite[Chap.~9.4]{Gro01}).

If we take the inner product of $\Theta_\L f$ with $f$ we arrive at
$\sum_{\l \in \L} |\langle f, \rho(\l) \varphi \rangle|^2$, which is the term we try to bound in~\eqref{eq:frame}. 
Lower bounds of~\eqref{eq:frame} are related to the invertibility of the Gabor frame operator, while upper bounds of~\eqref{eq:frame} ensure continuity of the Gabor frame operator, see Section~\ref{s:gabor}.

\section{Notation and tools}\label{s:notation}
We denote the Hilbert space of square integrable (complex-valued) functions on $\Rd$ by $\Lt$. The unitary operators $T_x$ of translation by $x \in \Rd$ and $M_\omega$ of modulation by $\omega \in \Rd$ act on functions by the rules
\begin{equation}\label{eq:Tx_Mw}
	T_x f(t) = f(t-x)
	\quad \text{ and } \quad
	M_\omega f(t) = e^{2 \pi i \omega \cdot t} f(t).
\end{equation}
The usual positive-semidefinite order between operators or matrices is denoted by $\preccurlyeq$.

The dot $\cdot$ denotes the Euclidean inner product on $\Rd$, which we see as a space of column vectors. The composition of $T_x$ and $M_\omega$ is called a time-frequency shift and is denoted by
\begin{equation}\label{eq:TFshift}
	\pi(z) = M_\omega T_x, \quad z=(x,\omega) \in \Rd \times \Rd \cong \R^{2d}.
\end{equation}
We say $\R^{2d}$ is the \textit{time-frequency plane} and also call it \textit{phase space} in our setting. The operators $T_x$ and $M_\omega$ do not commute in general but satisfy the following commutation relations:
\begin{equation}\label{eq:comm_relations}
	M_\omega T_x = e^{2 \pi i \omega \cdot x} \, T_x M_\omega.
\end{equation}
Hence, the set of time-frequency shifts is not closed under composition. Indeed, we have
\begin{equation}
	\pi(z)\pi(z') = e^{-2 \pi i x \cdot \omega'} \, \pi(z+z'), \qquad z=(x,\omega), \ z' = (x', \omega').
\end{equation}
We use the unitary Fourier transform $\F$ as usual in harmonic analysis \cite{Fol89} and time-frequency analysis~\cite{Gro01}:
\begin{equation}
    \F f(\omega) = \int_{\Rd} f(t) e^{-2 \pi i \omega \cdot t} \, dt.
\end{equation}
This formula is valid for a suitable dense subspace of $\Lt$ and $\F$ extends to a unitary operator on all of $\Lt$. Fixing a function $g \neq 0$, called window function, the short-time Fourier transform (STFT) of a function $f$ with respect to $g$ is given by
\begin{equation}\label{stft}
    V_g f(x,\omega) = \int_{\R^d} f(t) \overline{g(t-x)} e^{-2\pi i t \cdot \omega} dt = \langle f, M_\omega T_x g \rangle, 
    \quad \text{for $x,\omega \in \R^d$.}
\end{equation}
It plays a central role in the formulation of quantum paving as we will outline in the upcoming subsection.

A lattice (which we always assume to be full-rank) is a discrete co-compact subgroup of $\Rd$. We can also write it as $\Z$-module of a basis $M = \{v_1, \ldots , v_d\}$ of $\Rd$;
\begin{equation}
    \L = M \Z^d = \left\lbrace \sum_{l=1}^d k_l v_l \mid k_l \in \Z \right\rbrace, \quad M \in GL(d,\R).
\end{equation}
The dual lattice of $\L$ is denoted by $\L^\perp$ and characterized by
\begin{equation}
    \L^\perp = \{ \l^\perp \in \Rd \mid \l \cdot \l^\perp \in \Z, \ \forall \l \in \L\}.
\end{equation}
With this notation, the \textit{Poisson summation formula} reads
\begin{equation}\label{eq:PSF}
    \sum_{\l \in \L} f(\l+x) = \frac{1}{\vol(\Rd/\L)} \sum_{\l^\perp \in \L^\perp} \F f(\l^\perp) e^{2 \pi i z \cdot \l^\perp}.
\end{equation}
Note that $\L^\perp$ is indeed a lattice and we have $\L^\perp = M^{-T} \Z^d$, where $M^{-T}$ is the inverse transpose of $M$ and $\L = M \Z^d$. The density of $\Lambda$ is given by $\den(\Lambda)=1/|\det(M)| = 1/\vol(\Rd/\L)$.

\subsection{Gabor systems and frames}\label{s:gabor}

Gabor frames are a key concept of time-frequency analysis which enjoy a wide range of applications in mathematics, engineering, and signal processing. We refer to~\cite{Gro01} for an excellent introduction to Gabor analysis.  Classically, Gabor frames have been defined only for rectangular time-frequency lattices. But for obvious reasons, in this paper we prefer the flexibility that comes with general lattices in the time-frequency domain. Hence, 
given a non-zero window function $g \in L^2 (\R^d)$ and a lattice $\Lambda$ we define a Gabor system $\G(g,\L) = \{g_\lambda\}_{\lambda \in\Lambda}$ via
$$g_\lambda = \pi(\lambda) g , \qquad \lambda \in \Lambda.$$
If there exist finite constants $0 < A \le B$ such that for all $f \in L_2(\R^d)$ we have
\begin{equation}
    \label{gaborframe}
    A \|f\|_2 \le \sum_{\lambda \in \Lambda} | \langle f, \glam \rangle |^2 \le B \|f\|_2,
\end{equation}
then $\{g_\lambda\}_{\lambda \in\Lambda}$ is called a Gabor frame for $L^2(\R^d)$. The smallest constant $A$ and the largest constant $B$ in~\eqref{gaborframe} are called the frame bounds.
The associated frame operator $S$ is defined by
$$Sf = \sum_{\lambda  \in \Lambda} \langle f, \glam \rangle \glam.$$ It is easy to see that this positive-definite self-adjoint operator satisfies $I A \preccurlyeq S \preccurlyeq I B$, where $I$ is the identity operator on $L^2(\R^d)$. If in addition to~\eqref{gaborframe} it holds that the elements are linearly independent then $\{g_\lambda\}_{\lambda \in\Lambda}$ is called a Gabor Riesz basis.

For a given Gabor frame $\{g_\lambda, \lambda \in \Lambda\}$ with frame bounds $A,B >0$, $\{S^{-1} g_\lambda, \lambda \in \Lambda\}$ is a frame with frame bounds $B^{-1}, A^{-1} >0$, the so-called canonical dual Gabor frame. Since the frame operator commutes with time-frequency shift $\pi(\lambda)$ for $\lambda \in \Lambda$, we have
$S^{-1} g_\lambda = \pi(\lambda) S^{-1} g$, where the function $S^{-1}g$ is also referred to as the canonical dual window.

Also, for a Gabor system 
$\{g_\lambda, \lambda \in \Lambda\}$ (no matter if it is a frame or not), we define the associated Gram matrix $G$ by $G_{\lambda,\nu} = \langle \pi(\nu)g, \pi(\lambda) g \rangle$ for $\lambda,\nu \in \Lambda$. The structure of the Gram matrix $G$ was exploited by Janssen \cite{Jan96} to give first formulas for exact frame bounds for several window functions, including the Gaussian $\varphi$, under the assumption that the lattice is rectangular and its density is an integer. The integer density assumption is due to technical issues caused by the non-commutativity of time-frequency shifts.

An important quantity of the frame operator, both from a theoretical viewpoint as well as for numerical and practical purposes, is the condition number of $S$, given by the ratio $\frac{B}{A}$. Like for matrices, the motto ``smaller is better'' is in general true for frame operator condition numbers. In particular, if $A=B$, we have a {\em tight frame} and $S$ is just a multiple of the identity operator, making its inversion trivial. In that case the frame and its dual coincide up to normalization.

Throughout this paper we are only interested in Gabor frames where the window function is a Gaussian, i.e., $g = \varphi$, where $\phi$ is defined in~\eqref{eq:Gaussian}. Gaussian functions play a special role in quantum physics (as well as in many other areas), since they minimize the uncertainty principle~\cite{Gro01}. In quantum
mechanics and in quantum optics, time-frequency shifts of a Gaussian, play an important role, where they are usually referred to by the name {\em coherent states}~\cite{Klauder,Antonio}.

It is a rather deep result that $\{\varphi_\lambda, \lambda \in \Lambda\}$ is a frame for $L^2(\R)$ if and only if $\den(\Lambda) >1$, see~\cite{Gro01}. If we consider Gabor systems for $\Lt$ with $d > 1$, then $\delta(\L) > 1$ is a necessary density condition, but provably not always sufficient. A characterization of multivariate (Gaussian) Gabor frames is still missing.

\subsection{Symplectic methods}\label{s:symplectic}
We will exploit the symplectic structure of phase space and introduce the \textit{standard symplectic matrix} as
\begin{equation}
    \mathcal{J} = \begin{pmatrix*}[r]
        \mathbf{0} & \mathbf{I}\\
        -\mathbf{I} & \mathbf{0}
    \end{pmatrix*},
\end{equation}
where the blocks are of size $d \times d$ and denote the identity matrix $\mathbf{I}$ and the zero matrix~$\mathbf{0}$. For $z,z' \in \R^{2d}$, the \textit{standard symplectic form} is given by
\begin{equation}\label{eq:symplectic_form}
    \sigma(z,z') = z \cdot \mathcal{J} z'.
\end{equation}
A matrix $S \in SL(2d,\R)$ belongs to the symplectic group $Sp(d)$ if and only if
\begin{equation}\label{eq:symplectic}
    \sigma(z,z') = \sigma(S z, Sz')
    \quad \Longleftrightarrow \quad
    S^T \mathcal{J} S = \mathcal{J}.
\end{equation}
We recall that the symplectic group $Sp(d)$ is in general a proper subgroup of $SL(2d,\R)$. Only in the case $d=1$, we have $Sp(1) = SL(2,\R)$.

As we are only interested in even-dimensional lattices we also wish to assign them a \textit{symplectic dual} or \textit{adjoint} lattice $\L^\circ$. This may be characterized by commutations of time-frequency shifts or, equivalently, by using the standard symplectic form $\sigma(. \, , \, .)$:
\begin{align}\label{eq:adjoint_lattice}
	\L^\circ
    & = \{ \l^\circ \in \R^{2d} \mid \pi(\l)\pi(\l^\circ) = \pi(\l^\circ) \pi(\l), \, \forall \l \in \L \}\\
    & =
    \{ \l^\circ \in \R^{2d} \mid \sigma(\l,\l^\circ) \in \Z, \, \forall \l \in \L\}.
\end{align}
Note the similarity between the second characterization of the symplectic dual lattice $\L^\circ$ to the definition of the classical dual lattice $\Lambda^\perp$.

A \textit{symplectic lattice} in $\R^{2d}$ is a lattice $\Lambda = S \Z^{2d}$ with $S \in Sp(d)$. More generally, we will also call scaled versions $\alpha \Lambda = \alpha S \Z^{2d}$, $\alpha > 0$, symplectic if the generating matrix $S$ is symplectic. It follows from \eqref{eq:symplectic} that, for $S \in Sp(d)$,
\begin{equation}
    S = \mathcal{J}^{-1} S^{-T} \mathcal{J}
    \quad \text{ and } \quad
    \Lambda = S \Z^{2d} = \mathcal{J}^{-1} S^{-T} \mathcal{J} \Z^{2d} = \mathcal{J}^{-1} S^{-T} \Z^{2d} = \mathcal{J}^{-1} \Lambda^\perp = \mathcal{J} \L^\perp,
\end{equation}
as $\mathcal{J}^{-1} = - \mathcal{J}$ and $\Lambda^\perp = - \Lambda^\perp$. We also demand a \textit{symplectic Fourier transform}, which for a suitable function $F$ on $\R^{2d}$ is given by
\begin{equation}
    \F_\sigma F (z) = \iint_{\R^{2d}} F(z') e^{-2 \pi i \sigma(z',z)} \, dz' = \F F(\mathcal{J} z).
\end{equation}
Hence, $\F_\sigma$ inherits properties from the usual Fourier transform $\F$, e.g., it extends to a unitary operator on $\Lt[2d]$. In addition, it is involutive, i.e., $\F_\sigma^{-1} = \F_\sigma$. The symplectic Fourier transform $\F_\sigma$ and the symplectic dual lattice together provide us with a \textit{symplectic Poisson summation formula} (cf.~\cite[Sec.~2.2]{Fau18}):
\begin{equation}\label{eq:symplectic_PSF}
    \sum_{\lambda \in \Lambda} F(\lambda+z) = \sum_{\lambda \in \Lambda} \F_\sigma F(\lambda) e^{2 \pi i \sigma(z,\lambda)}.
\end{equation}
We refer to the textbooks by Folland \cite{Fol89} and Gröchenig \cite{Gos11} as well as the articles \cite[Sec.~1.1]{BetFau23}, \cite[App.~B]{BetFauSte21}, \cite{Fau18} for more details on symplectic methods.

As an application of the symplectic Poisson summation formula, we derive the Janssen representation of the Gaussian Gabor frame operator, which is given by
\begin{equation}\label{eq:Janssen_rep}
	S_{\L} = \vol(\L)^{-1} \sum_{\l^\circ \in \L^\circ} \langle \varphi, \pi(\l^\circ) \varphi \rangle \, \pi(\l^\circ).
\end{equation}
We see that, by the Janssen representation \eqref{eq:Janssen_rep}, the scaled frame operator $\vol(\L) S_\L$ is actually the quantum theta $\Theta_{\L^\circ}$ defined by \eqref{eq:quantum_theta}. If the lattice $\L$ is symplectic (up to scaling by $\alpha > 0$), then the adjoint lattice $\L^\circ$ can be replaced by the lattice $\L$ (appropriately scaled by $1/\alpha$).

\subsection{The Fundamental Identity of Gabor Analysis}

In what follows, it is convenient for us to introduce the \textit{modulation spaces} $M^p(\Rd)$, in particular the modulation space $M^1(\Rd)$, which is also known as Feichtinger's algebra $S_0(\Rd)$.
\begin{definition}[Modulation space]\label{def:Mp}
    The modulation space $M^p(\R)$ consists of all tempered distributions $f \in \mathcal{S}'(\R)$ such that
    \begin{equation}
        \norm{f}_{M^p(\R)} = \norm{V_{\varphi}f}_{L^p(\R^2)} < \infty.
    \end{equation}
\end{definition}
Note that $V_g g \in M^1(\R^2)$ if (and only if) $g \in M^1(\R)$ \cite[Cor.~5.5]{Jak18}. Hence, Feichtinger's algebra is a convenient setting for us, as it contains $\varphi$. Therefore, all manipulations in the sequel are justified and we do not run into any convergence issues (all series we encounter converge at least unconditionally).

The following (auxiliary) results we are going to present are well-known in Gabor analysis, but for the sake of completeness we provide the proof. We start with a result which is known as the \textit{Fundamental Identity of Gabor Analysis} (FIGA). It was first obtained by Janssen \cite{Jan98} for rectangular lattices and for general lattices it was shown by Feichtinger and Luef \cite{FeiLue06}. We also refer to the work of Gr\"ochenig and Koppensteiner \cite{GroKop19}, where all technical details are clarified. The proof uses the symplectic Poisson summation formula \eqref{eq:PSF} and the following orthogonality relations:
\begin{equation}\label{eq:OR}
         \langle V_{g_1} f_1, V_{g_2} f_2 \rangle_{L^2(\R^2)} = \langle f_1, f_2 \rangle_{L^2(\R)} \overline{\langle g_1, g_2 \rangle}_{L^2(\R)}.
\end{equation}
The proof of \eqref{eq:OR} is by straight forward computation and the interested reader is referred to the textbook of Gr\"ochenig \cite[Chap.~3.2]{Gro01}. Formula \eqref{eq:OR} is the STFT-version of Moyal's identity (see \eqref{eq:Moyal} below) for the Wigner distribution introduced in \cite{Moy49} (see also \cite[Chap.~6]{Gos17}, \cite[Chap.~4.3]{Gro01}).
\begin{proposition}[Fundamental Identity of Gabor Analysis]
    Assume we have functions $f_1, f_2, g_1, g_2 \in M^1(\R)$, then the following identity holds
    \begin{equation}\tag{FIGA}\label{eq:FIGA}
        \sum_{\l \in \L} V_{g_1} f_1(\l) \, \overline{V_{g_2} f_2 (\l)}
        = \frac{1}{\vol(\R^{2d}/\L)} \sum_{\l^\circ \in \L^\circ} V_{g_1} g_2(\l^\circ) \, \overline{V_{f_1} f_2(\l^\circ)}.
    \end{equation}
\end{proposition}
\begin{proof}
    As we assume all functions to be in the modulation space $M^1(\R)$, all of the following manipulations are justified. We observe that by the commutation relation \eqref{eq:comm_relations} we have
    \begin{equation}
        V_{\pi(\xi)g}(\pi(\xi)f)(z) = V_g f(z) e^{-2 \pi i \sigma (\xi,z)}.
    \end{equation}
    As it may easily be overlooked in the following calculation, we note that we use the fact that $\sigma(z,\xi) = - \sigma(\xi,z)$ and so $\overline{e^{-2 \pi i \sigma(\xi,z)}} = e^{-2 \pi i \sigma(z,\xi)}$. Hence, we have
    \begin{align}
        \mathcal{F}_\sigma ( V_{g_1}f_1 \ V_{g_2} f_2)(\xi) & = \iint_{\R^2} V_{g_1} f_1 (z) \overline{V_{g_2} f_2(z)} e^{-2 \pi i \sigma(z, \xi)} \, d z\\
        & = \iint_{\R^2}  V_{g_1}f_1(z) \overline{V_{\pi(\xi)g_2}(\pi(\xi)f_2)(z)} \, d z\\
        & = \langle V_{g_1} f_1, V_{\pi(\xi) g_2} (\pi(\xi)f_2) \rangle_{L^2(\R^2)}\\
        & = \langle f_1, \pi(\xi) f_2 \rangle_{L^2(\R)} \ \overline{\langle g_1, \pi(\xi) g_2 \rangle}_{L^2(\R)}
        = V_{f_2} f_1(\xi) \ \overline{V_{g_2} g_1(\xi)}.
    \end{align}
    Equation \eqref{eq:FIGA} now follows by applying the symplectic Poisson summation formula \eqref{eq:symplectic_PSF}:
    \begin{align}
        \sum_{\l \in \L} V_{g_1} f_1(\l) \, \overline{V_{g_2} f_2 (\l)}
        & = \vol(\L)^{-1} \sum_{\l^\circ \in \L^\circ} \mathcal{F}_\sigma \left(V_{g_1} f_1(\l) \, \overline{V_{g_2} f_2 (\l)}\right)\\
        & = \vol(\L)^{-1} \sum_{\l^\circ \in \L^\circ} V_{f_1} f_2(\l^\circ) \, \overline{V_{g_1} g_2(\l^\circ)}
    \end{align}
\end{proof}
We continue with the following observation, which is a slight variation of \cite[Lem.~4.2.1]{Gro01} in the textbook of Gr\"ochenig. The proof employs the Cauchy-Schwarz inequality, and when equality can be achieved.
\begin{lemma}\label{lem:Vgg}
    For non-vanishing $g \in L^2(\R)$ we note that $V_g g(0,0) \in \R_+$. Moreover, we have
    \begin{equation}
        V_g g(0,0) = \norm{g}_{L^2(\R)}^2
    \end{equation}
    and 
    the following estimate holds for all $(x,\omega) \in \R^2 \setminus \{(0,0)\}$:
    \begin{equation}
        |V_g g(x,\omega)| < V_g g(0,0).
    \end{equation}
\end{lemma}
\begin{proof}
    Clearly, $V_g g(0,0) = \norm{g}_{L^2(\Rd)}^2 \in \R_+$ by definition. We use the Cauchy-Schwarz inequality and the fact that $\pi(z)$ is unitary to get
    \begin{equation}
        |\langle g, \pi(z) g \rangle| \leq \norm{g}_{L^2(\Rd)}^2.
    \end{equation}
    Equality holds if and only if $\pi(z) g = c g$, $|c|=1$. For $z=(x,\omega)$, if $x \neq 0$, then $T_x |g| = |c \, \pi(z) \, g| = |c \, g| = |g|$ implies that $g$ is periodic, contradicting $g \in \Lt$. As $\F T_\omega = M_{-\omega} \F$ and $\F$ is unitary, we get that $\F g$ would need to be periodic in order to achieve equality if $\omega \neq 0$. Hence, the result follows.
\end{proof}
The following notion was introduced in \cite{TolOrr95}. A function $g$ satisfies \textit{Condition A} if
\begin{equation}\label{eq:cond_A}
    \sum_{\l^\circ \in \L^\circ} | V_g g (\l^\circ) | < \infty. \tag{Condition A}
\end{equation}
The next result was shown by Tolimieri and Orr \cite{TolOrr95} under the assumption that the window function $g$ indeed satisfies \eqref{eq:cond_A}. We only add the convenient choice that we have $g \in M^1(\R)$, so that we can directly use the result.
\begin{lemma}\label{lem:Vgg_bound}
    Let $g \in M^1(\R)$, then it satisfies \eqref{eq:cond_A} and, consequently, the Gabor system $\mathcal{G}(g,\L)$ has a finite upper frame bound $B$, which is at most
    \begin{equation}\label{eq:B}
        B \leq \vol(\L)^{-1} \sum_{\l^\circ \in \L^\circ} |V_g g(\l^\circ)| < \infty.
    \end{equation}
\end{lemma}
\begin{proof}
    The proof follows from Lemma \ref{lem:Vgg} and \eqref{eq:FIGA}. We compute
    \begin{align}
        \sum_{\l \in \L} |\langle f, \pi(\l) g \rangle|^2 
        & = \sum_{\l \in \L} V_{g} f(\l) \overline{V_g f(\l)}
        = \vol(\L)^{-1} \sum_{\l^\circ \in \L^\circ} V_f f(\l^\circ) \overline{V_g g(\l^\circ)}\\
        & \leq \vol(\L)^{-1} \sum_{\l^\circ \in \L^\circ} |V_f f(\l^\circ)| |V_g g(\l^\circ)|
        \leq \vol(\L)^{-1} \sum_{\l^\circ \in \L^\circ} |V_f f(0)| |V_g g(\l^\circ)|\\
        & = \left(\vol(\L)^{-1} \sum_{\l^\circ \in \L^\circ}|V_g g(\l^\circ)| \right) \norm{f}_{L^2(\R)}^2.
    \end{align}
    Avoiding technicalities ($f$ is not necessarily in $M^1(\R)$), which are however clarified in \cite{FeiLue06} or \cite{GroKop19}, we see that the frame inequality \eqref{eq:frame} if fulfilled for $f \in L^2(\R)$ with an upper frame bound given by \eqref{eq:B}. The bound is indeed finite as $g \in M^1(\R)$.
\end{proof}
The last result we present in this section is the \textit{Janssen representation} of the frame operator, which shows that it is the quantum theta $\vol(\L)^{-1} \Theta_{\L^\circ}$.
\begin{lemma}[Janssen representation]
    For a window $g \in M^1(\R)$ and a lattice $\L$, the frame operator $S = S_{g,\L}$ can be expressed as
    \begin{equation}
        S = \vol(\L)^{-1} \sum_{\l^\circ \in \L^\circ} \langle g, \pi(\l^\circ) g \rangle \pi(\l^\circ).
    \end{equation}
\end{lemma}
\begin{proof}
    The result is well-known and it follows basically from \eqref{eq:FIGA}. We compute
    \begin{equation}
        \langle S f, h \rangle = \sum_{l \in \L} V_g f(\l) \ \overline{V_g h(\l)} = \vol(\L)^{-1} \sum_{\l^\circ \in \L^\circ} \langle g, \pi(\l^\circ) g \rangle \langle \pi(\l) f, h \rangle,
    \end{equation}
    where we used \eqref{eq:FIGA}, the fact that $V_h f(z) = e^{-2 \pi i x \omega} \overline{V_f h(-z)}$, $z=(x,\omega)$ and that $\L$ is a lattice, which implies that $\l \in \L$ if and only if $-\l \in \L$. Hence, we obtain the result in the weak operator topology.
\end{proof}

An important consequence of Janssen's representation is the following duality principle
\begin{lemma}[Ron-Shen duality principle]\label{duality}
Let $g \in L^2(\R^d)$. Then the Gabor system $\{g_\lambda, \lambda \in \Lambda\}$
is a frame for $L^2(\R^d)$ if and only if $\{g_{\lambda^\circ}, \lambda^\circ \in \Lambda^\circ\}$ is a Riesz basis
for its closed linear span.
\end{lemma}
In other words,  $S$ is the frame operator for 
$\{g_\lambda, \lambda \in \Lambda\}$ with frame bounds $A,B>0$ if and only if the Gram matrix $G^\circ$ for $\{g_{\lambda^\circ}, \lambda^\circ \in \Lambda^\circ\}$ satisfies $ AI \preccurlyeq G^\circ \preccurlyeq B I$.
An easy but useful consequence of Lemma~\ref{duality} is the following basic observation: The family $\{\gamma_\lambda, \lambda \in \Lambda\}$ is a dual Gabor frame for 
$\{g_\lambda, \lambda \in \Lambda\}$ if and only if 
$\{\gamma_{\lambda^\circ}, \lambda^\circ \in \Lambda^\circ\}$ is a dual Gabor Riesz basis for $\{g_{\lambda^\circ}, \lambda^\circ \in \Lambda^\circ\}$.

\section{Quantum pavings and wireless communications}\label{s:ofdm}

The problem of quantum pavings arises also in non-quantum settings. One such example occurs in modern wireless communications~\cite{Liu2020,Liya,StrBea03}.
Transmission over wireless channels is subject to time dispersion due to
multipath propagation as well as frequency dispersion due to the Doppler
effect~\cite{Molisch}. While Orthogonal Frequency Division Multiplexing (OFDM) is currently the communication system of choice, it is expected that future wireless communication standards such as 6G will also feature {\em pulse-shaping OFDM}, cf. e.g.~\cite{Liu2020,Liya}. Skipping engineering details, from a mathematical viewpoint there are deep connections between such a pulse-shaping OFDM system and Gabor Riesz bases (the latter are, via Lemma~\ref{duality}, dual to Gabor frames). Here, the pulse shape plays the role of the Gabor window. 

Indeed, the transmission functions of a standard pulse shaping OFDM system consist of translations and modulations of a single pulse shape $f$, i.e.,
\begin{equation}
\label{ofdmsystem}
g_{kl}(t)=g(t-kT) e^{2\pi i lFt}, 
\end{equation}
for some $T,F>0, k,l \in \Z$, where
 $T$ is called the symbol period and $F$ is the carrier separation, see~\cite{Bolcskei,KozekMolisch,StrBea03} for details.
The functions in~\eqref{ofdmsystem} are identical to a Gabor system\footnote{In reality the index range for $k$ and $l$ in~\eqref{ofdmsystem} is large but finite. However, this difference is negligible  for our analysis.}.

Let $c_{k,l}, k,l\in \Z$  be the (already block-coded)
data symbols to be transmitted. The OFDM transmission signal is then given by
\begin{equation}
\label{OFDMsignal}
s(t)= \sum_{k,l} c_{k,l} g_{k,l}.
\end{equation}
The received signal is (ignoring additive noise)
\begin{equation}
r(t) = (H s)(t)
= \int \limits_{-\infty}^{+\infty}  h(t,\tau) s(t-\tau) d\tau,
\label{eq5}
\end{equation}
where $h(t,\tau)$ is the impulse response of the time-varying channel
$H$.
The data symbol $c_{k,l}$ is transmitted at the $(k,l)$-th lattice point
of the rectangular time-frequency lattice $(kT,lF), k,l \in \Z$.

An important point is that in order to reconstruct the data $c_{k,l}$  from the received signal $r$, we require the transmission functions to be linear independent. This means we are dealing with Gabor Riesz bases (for a subspace of $L^2(\R)$)) instead of Gabor frames.

While the data $c_{kl}$ live in the discrete world, the signals $s$ and $r$ live in the continuous (analog) world. The received signal $r$ is transformed back into the discrete world by computing the inner product of $r$ with a dual Riesz basis $\{f_{k,l}\}$,
i.e., we obtain\footnote{The interchange of summation and integration in~\eqref{innprod} is justified under mild smoothness and decay conditions.}
\begin{equation}
d_{k,l}=\langle r ,f_{k,l} \rangle
=\langle H s ,f_{k,l} \rangle 
= \sum_{k',l'} c_{k',l'} \langle H g_{k',l'},
f_{k,l} \rangle.
\label{innprod}
\end{equation}
The process of reconstructing of the transmitted data $c_{k,l}$ from the received distorted data $d_{k,l}$ is called {\em equalization}.
In the ideal (but unrealistic) case that the wireless channel does not
introduce any distortion (i.e., $H=I$) we obtain
$d_{k,l} = c_{k,l}$ for all $k,l$,
owing to the biorthogonality property $\langle g_{k',l'} , f_{k,l} \rangle = \delta_{kl,k'l'}$.

However, in practice wireless channels are time-dispersive and
frequency-dispersive. Hence when the transmission pulses $g_{kl}$ pass through the 
channel $H$, they get spread out in time and frequency, causing
interference in the transmitted data.  This
interference can be reduced if the pulse shape is well-localized in
time and frequency. 
Hence, a natural choice is to set $g=\phi$,~\cite{KozekMolisch}.

To further improve the stability of the pulse-shaping OFDM system against
channel interference we have to increase the distance between adjacent data symbols by increasing the distance between adjacent grid points.
However increasing $T$ and/or $F$ results in an undesirable loss of data rate, which is captured by the quantity $\frac{1}{TF}$. Hence we are concerned with the following question, whether we can increase the distance between adjacent grid points
without reducing the data rate?

The above considerations motivate the introduction of Lattice-OFDM which is based on general time-frequency lattices, see~\cite{StrBea03}. Thus, our transmission functions  now take the form 
$\{\phi_\lambda, \lambda \in \Lambda\}$, where $\Lambda$ is a lattice in $\R^2$. But which lattice shall we choose? Naturally, the hexagonal lattice comes to mind.
See Figure~\ref{fig:rect} for a visual demonstration of the benefits of the hexagonal lattice over the rectangular lattice. But is the hexagonal lattice the best choice?

\begin{figure}[ht]
\begin{center}
\includegraphics[width=0.475\textwidth]{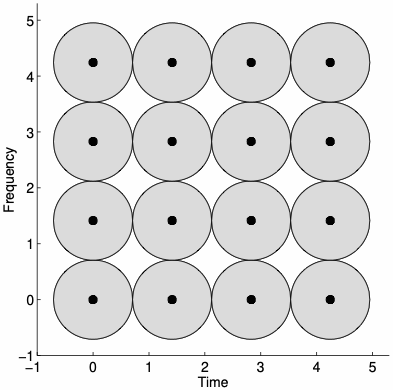}
\hspace{\fill}
\includegraphics[width=0.475\textwidth]{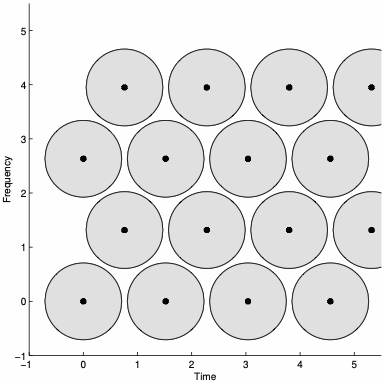}
\caption{A pulse-shaping OFDM system using a Gaussian pulse with rectangular lattice (left) and hexagonal lattice (right). The data symbols $c_{k,l}$ are transmitted at the lattice points $\lambda \in \Lambda$. The spheres represent the effective support $\supp_\eps  (F) = \{z \mid F(z) > \varepsilon\}$ of the STFT $V_{g} \pi(\lambda)g, \lambda \in \Lambda$ for both the square and the hexagonal lattice, respectively. The choice of $\eps$ is such that the difference between the two lattices becomes apparent. The use of the hexagonal lattice clearly increases the distance between adjacent spheres as compared to the square lattice, thus leading to reduced interference for time-frequency dispersive channels.}
\label{fig:rect}
\end{center}
\end{figure}

We can rephrase this question  also as follows: For a lattice $\Lambda$ of fixed density, we consider the associated Gram matrix $G=G(\Lambda)$ with entries 
$G_{\lambda,\nu} = \langle \pi(\nu) \phi, \pi(\lambda) \phi \rangle$, where $ \lambda, \nu \in \Lambda$. Based on our normalization of $\phi$, we want to find a lattice $\Lambda$ such that 
 the off-diagonal entries of $G(\Lambda)$ are as close to zero as possible. Heuristically, this should also manifest itself in a smaller condition number of $G$.

Furthermore, for a given Gabor frame $\{\phi_\lambda, \lambda \in \Lambda\}$, the dual window $\psi =S^{-1}\phi$ has exponential decay in time and frequency~\cite{Jan96,Str01}\footnote{Using techniques from~\cite{GroStr07} it is not difficult to extend the decay results in~\cite{Jan96,Str01} for Gabor frames defined for rectangular lattices to general lattices.}, where the constant in the exponent depends on the condition number of the frame operator $S$. Thus, in order to optimize the time-frequency localization of the dual window $\psi$ we can now ask, which lattice $\Lambda$ of fixed density minimizes the condition number of the frame operator $S$ associated with
$\{\phi_\lambda, \lambda \in \Lambda\}$? 

The above considerations lead us exactly to the quantum paving problem for $d=1$ (i.e., phase space dimension 2). 
In the next section we will dive deeper into the connections between quantum paving and sphere packings and coverings.

\section{Packing and covering: why and how?}

We will now explain how classical sphere packing and covering lead to quantum packing, quantum covering and quantum paving. The problems stem from M.F.'s study and collaborative successful attempt with L.~Bétermin and S.~Steinerberger to find an optimal sampling strategy for Gaussian Gabor systems in $\Lt[]$ \cite{BetFauSte21} (see also \cite{Fau18}, \cite{FauSte17}). Optimality is measured in terms of the condition number of the frame operator, i.e., the ratio $B_\L/A_\L$.

During a workshop at the ESI in Vienna more than 20 years ago, it was suggested by T.S.\ that optimal sphere packings could be of relevance for this task. Indeed, the conjecture raised in \cite{StrBea03} (see also \cite{AbrDoe12, StrHea03}) was that a hexagonal lattice should minimize the frame condition number $B_\L/A_\L$. A key insight by M.F., twenty years later, was that optimal sphere coverings should also be considered, leading to the collaborative solution of M.F.\ with L.~Bétermin and S.~Steinerberger \cite{BetFauSte21}.

So, how do optimal sphere packings and coverings meet Gabor frames? We may relax \eqref{eq:frame} by only considering elements from the set $g = \{\pi(z) \varphi \mid z \in \R^{2d} \}$ of time-frequency shifted Gaussians.
A brief computation shows that $|\langle \pi(z) \varphi, \pi(\l) \varphi \rangle|^2 = e^{-\pi (\l-z)^2}$ (see \cite[Chap.~1.5]{Gro01}).
\begin{equation}\label{eq:relaxed}
	\text{For } f = \pi(z) \varphi \text{ the frame inequality \eqref{eq:frame} becomes: }
	\hspace*{1.cm}
	A_\L \leq \sum_{\l \in \L} e^{-\pi |\l-z|^2} \leq B_\L, \quad \forall z \in \R^{2d}.
\end{equation}
Inequality \eqref{eq:relaxed} only yields approximations to the exact constants $A_\L$ and $B_\L$, which are nonetheless quite accurate and sometimes sharp (see \cite{BetFauSte21, Jan96, TolOrr95}). For lattices, the constants in \eqref{eq:relaxed} ask for the optimal polarization constant $A_\L$ and the optimal energy minimization constant $B_\L$. Writing the lattice $\L$ as $\L = \alpha^{1/2} \L_0 \subset \R^{2d}$, where $\L_0$ has density 1 (so $\den(\L) = \alpha^{-d})$, \eqref{eq:relaxed} can be written as
\begin{equation}
    A_{\L_0} \leq \sum_{\l_0 \in \L_0} e^{- \pi \alpha |\l - z|^2} \leq B_{\L_0}, \quad \forall z \in \R^{2d}, \ \alpha > 0.
\end{equation}
The constants $A_{\L_0}$ and $B_{\L_0}$ now also depend on the scaling parameter $\alpha$. For $\alpha \to \infty$, the lattice maximizing $A_{\L_0}$ is the one with the best covering property \cite{BetFauSte21}, \cite{HarPetSaf22} and the lattice minimizing $B_{\L_0}$ needs to solve the sphere packing problem \cite{Coh-Via22}.
\begin{figure}[ht]
	\includegraphics[width=.475\textwidth]{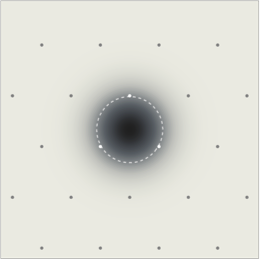}
	\hspace{\fill}
	\includegraphics[width=.475\textwidth]{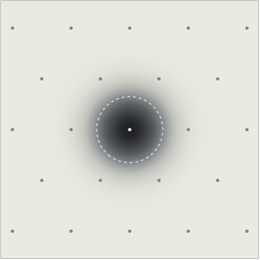}
	\caption{The function $|V_\varphi \varphi|^2$ in the time-frequency plane, centered at a deep hole (left) and a lattice point (right). Most of its mass is concentrated on a disk of area 1 (inside the dashed circle). The hexagonal lattice minimizes the distance to the deep hole (best covering) and maximizes the distance to other lattice points (best packing). Heuristically, this corresponds to optimizing $A$ and $B$, respectively.}
    \label{fig:gaussphasespace}
\end{figure}

Heuristically, having a Gaussian $\pi(z) \varphi$ centered near a \textit{deep hole}, i.e., a point globally maximizing the distance to its neighbors, will make $A_\L$ small. Having many lattice points close to each other results in redundant information and $B_\L$ will be large:
\begin{equation}
	\textbf{
		\boxed{
            \hspace{0.5cm}
			\text{maximize } A_\L \longleftrightarrow \text{ sphere covering}
			\hspace{3cm}
			\text{minimize } B_\L \longleftrightarrow \text{ sphere packing}.
            \hspace{0.5cm}
		}
	}
\end{equation}

Figure~\ref{fig:gaussphasespace} depicts the Gauss function in phase space via its spectrogram $|V_\varphi \varphi|^2$. Note: as $g$ is a \mbox{$d$-dimensional} Gaussian we have that $|V_g g|^2$ is a $2d$-dimensional Gaussian.  Most of the mass of $|V_\varphi \varphi|^2$ is concentrated inside a sphere of radius 1 (inside the dashed circle in Figure~\ref{fig:gaussphasespace}).
With some poetic license, we may refer to such a truncated time-frequency representation of the Gaussian as a {\em quantum sphere}, which lends further justification to the names {\em quantum packing} and {\em quantum covering}.



\section{Phase-space formalism}\label{s:phase}
The presented problems can also be formulated in phase space. For this, we introduce the (cross-)Wigner distribution which for functions $f, g \in \Lt$ is given by
\begin{equation}\label{eq:Wigner}
	W(f, g)(x, \omega) = \int_{\Rd} f(x+\tfrac{t}{2}) \overline{g(x-\tfrac{t}{2})} e^{-2 \pi i \omega \cdot t} \, dt, \quad t,x,\omega \in \Rd.
\end{equation}
In the case $f = g$ we simply write $Wf(x,\omega)$. This is often used as a quasi-probability distribution (as it fails to be positive in general) in quantum mechanics.
Applying $\F_\sigma$ to a (cross-)Wigner distribution leads to the ambiguity function, which is a symmetric version of the STFT. We have
\begin{equation}
	\F_\sigma W(f,g)(z) = A(f,g)(z).
\end{equation}
The function $A(f,g)$ is the (cross-)ambiguity function of $f$ and $g$, defined by
\begin{equation}
	A(f, g)(x, \omega) = \int_{\Rd} f(t-\tfrac{x}{2}) \overline{g(t-\tfrac{x}{2})} e^{-2 \pi i \omega \cdot t} \, dt = \langle f, \rho(x,\omega) g \rangle,
\end{equation}
where $\rho = T_{x/2} M_\omega T_{x/2} = M_{\omega/2} T_x M_{\omega/2}$ is a symmetric time-frequency shift, or, equivalently, $\rho$ is obtained from the Schrödinger representation of the Heisenberg group. The Heisenberg group $(\mathbf{H},\circ)$ is a non-commutative group, topologically isomorphic to $\Rd \times \Rd \times \R \cong \R^{2d+1}$, with group law
\begin{equation}
	\mathbf{h} \circ \mathbf{h}' = (x,\omega,\tau) \circ (x',\omega',\tau') = (x+x',\omega+\omega', \tau+\tau' + \tfrac{1}{2}(x \cdot \omega' - x' \cdot \omega)), \quad \mathbf{h}, \mathbf{h}' \in \mathbf{H}.
\end{equation}

We note that, by setting $z=(x,\omega)$ and $z'=(x',\omega')$ we may actually write
\begin{equation}
	\mathbf{h} \circ \mathbf{h}' = (z,\tau) \circ (z',\tau') = (z+z',\tau+\tau' + \tfrac{1}{2} \sigma(z,z')),
\end{equation}
where $\sigma(. \, , .)$ is the standard (skew-symmetric) symplectic form defined by \eqref{eq:symplectic_form} . The inverse of $\mathbf{h} \in \mathbf{H}$ is given by $\mathbf{h}^{-1} = (-z,-\tau)$. So, $(\mathbf{H}, \circ)$ is indeed a group, but non-commutative. Hence, algebraically it differs from $(\R^{2d+1},+)$. Its importance in TFA comes from its unitary representation(s) on $\Lt$. We recall that (in general) time-frequency shifts are not closed under composition due to \eqref{eq:comm_relations};
\begin{equation}
	\pi(z) \pi(z') = e^{-2 \pi i \omega' \cdot x} \pi(z+z').
\end{equation}

It is now advantageous to use the symmetric time-frequency shifts $\rho(z)$ and the (cross-)\textit{ambiguity function} $A(f,g)$ of two functions $f,g \in \Lt$ instead of $\pi(z)$ and $V_gf$:
\begin{equation}
	\rho(z) = M_{\omega/2}T_x M_{\omega/} = T_{x/2} M_\omega T_{x/2} = e^{-\pi i \omega \cdot x} \pi(z),
	\quad \text{ and } \quad
	A(f,g)(z) = \langle f, \rho(z) g \rangle = e^{\pi i \omega \cdot x} V_g f(z).
\end{equation}
Note that the unitary operators $\rho(z)$ are also not closed under composition. The appearance of a phase factor is unavoidable, but, by adding this unitary phase to the operator, we are able to obtain a group:
\begin{equation}
	e^{2 \pi i \tau} \rho(z) e^{2 \pi i \tau'} \rho(z') = e^{2 \pi i (\tau+\tau' + \frac{1}{2} \sigma(z,z'))} \rho(z+z'), \quad \tau \in \R.
\end{equation}
So, we see that by using the composition of the Heisenberg group, the unitary operators $\rho(\mathbf{h}) = \rho(z,\tau) = e^{2 \pi i \tau} \rho(z)$ also form a non-commutative group. It is customary to use the same symbol for the symmetric time-frequency shifts and the unitary representation of the Heisenberg group, as they only differ by the unitary phase factor $e^{2 \pi i t}$. The time-frequency shifts $\pi(z)$ acting on $\Lt$ yield a unitary representation of the \textit{polarized} Heisenberg group (see \cite[Chap.\ 9.1]{Gro01}).

After this short excursion into the Schrödinger representation theory of the Heisenberg group, we return to our phase space formalism of quantum packing, covering and paving. Recall Moyal's identity for the Wigner distribution, which for functions $f_1, \, f_2, \, g_1, \, g_2 \in \Lt$ is given by
\begin{equation}\label{eq:Moyal}
	\left\langle W(f_1,g_1), W(f_2, g_2) \right\rangle_{\Lt[2d]} = \langle f_1, f_2 \rangle_{\Lt} \, \overline{\langle g_1, g_2 \rangle}_{\Lt}.
\end{equation}
Note that the inner product on the left-hand side in \eqref{eq:Moyal} is in $\Lt[2d]$, whereas on the right-hand side they are taken in $\Lt$. Using Moyal's formula, an equivalent formulation of \eqref{eq:frame} in terms of Wigner distributions can be formulated. Consider the following double-sided inequality, for positive constants $A_\L$ and $B_\L$, satisfying
\begin{equation}\label{eq:frame_Wigner}
	A_\L \norm{W(f,\varphi)}_{\Lt[2d]}^2 \leq \sum_{\l \in \L} \left| \langle W(f, \varphi), \, W(\rho(\l) \varphi,\varphi) \rangle_{\Lt[2d]} \right|^2 \leq B_\L \norm{W(f, \varphi)}_{\Lt[2d]}^2,
\end{equation}
for all $f \in \Lt$. Now, pick $f_1 = f$, $f_2 = \rho(\l) \varphi$ and $g_1=g_2 = \varphi$ in \eqref{eq:Moyal} to obtain
\begin{equation}
	\sum_{\l \in \L} | \langle W(f, \varphi), W(\rho(\l) \varphi, \varphi) \rangle |^2 = \sum_{\l \in \L} | \langle f, \rho(\l) \varphi \rangle \, \underbrace{\langle \varphi, \varphi \rangle}_{=1} |^2 = \sum_{\l \in \L} |\langle f, \rho(\l) \varphi \rangle |^2.
\end{equation}
It remains to observe, by using Moyal's identity \eqref{eq:Moyal} again, that
\begin{equation}
	\norm{W(f, \varphi)}_{\Lt[2d]} = \norm{f}_{\Lt} \norm{\varphi}_{\Lt} = \norm{f}_{\Lt}.
\end{equation}
Hence, for all $f \in \Lt$ it holds that
\begin{align}
	& \quad &
    A_\L \norm{W(f,\varphi)}_{\Lt[2d]}^2
    & \leq \sum_{\l \in \L} \left| \langle W(f, \varphi), \, W(\rho(\l) \varphi,\varphi) \rangle \right|^2  \leq B_\L \norm{W(f, \varphi)}_{\Lt[2d]}^2\\
	\Longleftrightarrow & \quad &
    A_\L \norm{f}_{\Lt}^2
    & \leq \sum_{\l \in \L} \left| \langle f, \rho(\l) \varphi) \rangle \right|^2 \leq B_\L \norm{f}_{\Lt}^2,
\end{align}
with same optimal constants $A_\L$ and $B_\L$. Hence, problems \ref{pro:qpacking}, \ref{pro:qcovering}, \ref{pro:qpaving} have an equivalent formulation using phase space formalism. Since $\F_\sigma$ is a unitary operator, we have Parseval's identity and Plancherel's formula at hand and also have the following equivalent two-sided inequalities:
\begin{align}
    & \quad &
    A_\L \norm{f}_{\Lt}^2
    & \leq \sum_{\l \in \L} \left| \langle f, \rho(\l) \varphi) \rangle \right|^2 \leq B_\L \norm{f}_{\Lt}^2\\
    \Longleftrightarrow & \quad &
    A_\L \norm{A(f,\varphi)}_{\Lt[2d]}^2
    & \leq \sum_{\l \in \L} \left| \langle A(f, \varphi), \, A(\rho(\l) \varphi,\varphi) \rangle \right|^2  \leq B_\L \norm{A(f, \varphi)}_{\Lt[2d]}^2
\end{align}
The same holds true if $A(f,\varphi)$ is replaced by the STFT $V_\varphi f$ or by the Rihacek distribution, which is the symplectic Fourier transform of the STFT.

\section{Existing (auxiliary) results}


In this section we summarize some results on energy minimization and polarization which can be used to solve quantum packing, quantum covering and quantum paving in (quite) specific situations. Quite generally, energy minimization for lattices of a fixed density considers the task
\begin{equation}
    \text{solve:} \quad \min_{\L} \sum_{\l \in \L} p(|\l|),
\end{equation}
where $p(r) = f(r^2)$ and $f$ is completely monotone. The dual problem of polarization for lattices is
\begin{equation}
    \text{solve:} \quad \max_{\L} \min_{y} \sum_{\l \in \L} p(|\l+y|),
\end{equation}
where $p$ is as above. In both cases Gaussians are of special interest due to the Bernstein-Widder theorem.
\begin{thm}[Bernstein-Widder]\label{thm:Bernstein-Widder}
    Let $f$ be a completely monotone function, then there exists a non-negative finite Borel measure on $\R_+$ with distribution function $\mu_f$ such that
    \begin{equation}
        f(r) = \int_{\R_+} e^{-\alpha r} d \mu_f(\alpha).
    \end{equation}
\end{thm}
Hence, completely monotone functions of squared distance, i.e., $p(r) = f(r^2)$, $f$ completely monotone, are in the span of the Gaussians $g(\alpha r)$, $\alpha >0$ and the constant function (cf.~\cite{Coh-Via22}). For further reading we refer to the article by Cohn and Kumar \cite{CohKum07} or to \cite{Coh-Via22, FauShaZlo23, FauSte24}.

The first result we present is due to Montgomery \cite{Mon88} and concerns energy minimization in dimension~2.
\begin{thm}[Montgomery, 1988]
	Among all lattices of any fixed density and for any fixed $\alpha > 0$, the hexagonal lattice $\L_2 \subset \mathbb{R}^2$ is the unique solution the following problem.
\begin{equation}
     \text{Solve:} \quad \min_{\L} \sum_{\l \in \L}e^{- \pi \alpha |\l|^2}.
\end{equation}
\end{thm}
 The next result is from \cite{BetFauSte21} and can be regarded dual to the theorem of Montgomery. It solves, e.g., the polarization problem in dimension 2 as conjectured in \cite{HarPetSaf22}.

\begin{thm}[B\'etermin, Faulhuber, Steinerberger, 2021]
	Among all lattices of any fixed density and for any fixed $\alpha > 0$, the hexagonal lattice $\L_2 \subset \mathbb{R}^2$ is the unique solution the following problem.
\begin{equation}
     \text{Solve:} \quad \max_{\L} \min_{y} \sum_{\l \in \L}e^{- \pi \alpha |\l+y|^2}.
\end{equation}
\end{thm}

The last result we present certifies the so-called universal optimality of the $\mathsf{E}_8$ and Leech lattice in dimension 8 and 24, respectively.
\begin{thm}[Cohn, Kumar, Miller, Radchenko, Viazovska, 2022]
    Among all lattices of any fixed density and for any fixed $\alpha > 0$, the $\mathsf{E}_8$ and Leech lattice give the unique solution to the following problem in their respective dimension.
\begin{equation}
     \text{Solve:} \quad \min_{\L} \sum_{\l \in \L}e^{- \pi \alpha |\l|^2}.
\end{equation}
\end{thm}
We remark that the result hols for periodic point configurations and not only lattices, in contrast to the presented results in dimension 2. These, however, are also expected to hold for general point configurations.

Due to the appearance of the Gaussian function in all of the above results and the fact that
\begin{equation}
    |\pi(z) \varphi, \pi(\l) \varphi|^2 = e^{- \pi |\l-z|^2},
\end{equation}
we can make take advantages of these results to solve quantum packing, covering and paving in phase space dimension 2 under the assumption that the lattice density is even.

\section{Known solutions and some conjectures}\label{s:known}

The only case where we know solutions to the quantum packing, covering, and paving problems is for the Hilbert space $\Lt[]$ (i.e., $d=1$ and phase space dimension is 2) and only if we consider solely lattices where the density is an even integer. This is a consequence of combining the polarization result from \cite{BetFauSte21}, the energy minimization result from \cite{Mon88} (see also \cite{Fau18-JFAA}) and a result of Janssen on Gabor frame bounds~\cite{Jan96}.
\begin{thm}[B\'etermin, Faulhuber, Steinerberger, 2021]
	Let $d=1$, so phase space is $\R^2$. Consider the quantum problems \ref{pro:qpacking}, \ref{pro:qcovering}, \ref{pro:qpaving} above for lattices of even integer density. Then, for fixed density and within this class of lattices, the hexagonal lattice $\L_2$ solves all of the above quantum problems simultaneously.
\end{thm}
\begin{proof}
    We give the sketch of the proof and start with defining the (theta-like) functions in phase space:
    \begin{equation}
        F_\L(z) = \sum_{\l \in \L} e^{-\pi |\lambda+z|^2}
        \quad \text{ and } \quad
        \widehat{F}_{\L}(z) = \sum_{\l \in \L} e^{-\pi |\l|^2} e^{2 \pi i \sigma(\l, z)}.
    \end{equation}
	For a lattice $\L \subset \R^2$, the results of Janssen \cite{Jan96} on Gaussian Gabor systems (see \cite{Fau18} for the concrete formula we use here) shows that
	\begin{equation}
		A_\L = \min_{z \in \R} \vol(\L)^{-1} \, \widehat{F}_{\Lambda^\circ}(z)
        \quad \text{ and } \quad
		B_\L = \max_{z \in \R} \vol(\L)^{-1} \, \widehat{F}_{\Lambda^\circ}(z).
	\end{equation}
	A simple application of the triangle inequality yields
	\begin{align}
		\left| \vol(\L)^{-1} \sum_{\l^\circ \in \L^\circ} e^{-\pi |\l^\circ|^2} e^{2 \pi i \sigma(\l^\circ, z)} \right| & \leq  \vol(\L)^{-1} \sum_{\l^\circ \in \L^\circ} \left| e^{-\pi |\l^\circ|^2} \right| \, \left| e^{2 \pi i \sigma(\l^\circ, z)} \right|\\
		& = \vol(\L)^{-1} \sum_{\l^\circ \in \L^\circ} e^{-\pi |\l^\circ|^2},
	\end{align}
	which shows that the maximum is always taken at a lattice point $\l \in \L$ (in particular at the origin). Therefore, Montgomery's result \cite{Mon88} implies that $B_\L$ is minimal if and only if $\L$ is the hexagonal lattice $\L_2$. The polarization theorem in \cite{BetFauSte21} implies that $A_\L$ is maximal if and only if the lattice is hexagonal and, hence, the quantum paving problem is uniquely solved by the hexagonal lattice under the assumptions of the theorem.
\end{proof}

Note that finding the minimum of the function $F_\Lambda(z)$ or, equivalently $\vol(\L)^{-1} \widehat{F}_{\Lambda^\circ}(z)$ is a difficult task. There are only few rigorous results and generally the minimizer does not only depend on the geometry but also on the density of the lattice. In some exceptional cases, however, the minimizer is static: for $\Z^{2d}$ the minimizer of $f_{\Z^{2d}}$ is always taken at $(\Z+1/2)^{2d}$ and if $\Delta$ is a full-rank diagonal matrices, then the minimizer of $f_{\Delta \Z^{2d}}$ is found at $\Delta (\Z+1/2)^{2d}$ (see \cite{Jan96}, \cite{FauShaZlo23}). To the best of our knowledge, the only other lattice for which we know the minimizer of $f_\L(z)$ is the hexagonal lattice $\L_2$, where the minimum is taken at deep holes \cite{Bae97}.

It is quite remarkable that the recent result of L.~B\'etermin, M.~Faulhuber, and S.~Steinerberger \cite{BetFauSte21} shows that $\L_2$ maximizes the minimum among all $F_\L(z)$ for fixed density of $\L$, without knowing the minimizer of $F_\L(z)$ in general.

\begin{conjecture}\label{conj:hexagonal}
	Due to the blatant lack of counter examples, we conjecture that the hexagonal lattice solves all quantum problems simultaneously, among periodic configurations of arbitrary fixed density.
\end{conjecture}

In higher dimensions, the lack of knowledge on the characterization of Gabor frames makes it hard to get results on quantum covering and paving. However, for quantum packing we may use Lemma~\ref{lem:Vgg_bound} and the frame inequality \eqref{eq:frame} to obtain the following estimates on~$B_\L$.
\begin{equation}
    \sum_{\l \in \L} |V_\varphi \varphi(\l)|^2 \leq B_\L \leq \frac{1}{\vol(\R^{2d}/\L)} \sum_{\l^\circ \in \L^\circ} |V_\varphi \varphi(\l^\circ)| = \widetilde{B}_{\L}
\end{equation}

We conclude with a first partial result in higher dimensions and some more open problems.
\begin{theorem}
	Let $d\in\{4,12\}$, so the dimension of phase space is 8 or 24. Then the $\mathsf{E}_8$ lattice and the Leech lattice minimize $\widetilde{B}_\L$ in their respective dimensions.
\end{theorem}
\begin{proof}
	This follows as a consequence of the universal optimality of the $\mathsf{E}_8$ and Leech lattice \cite{Coh-Via22}.
\end{proof}

The following conjecture is based on what Robert Calderbank calls ``Proof by Intimitation'', which was ``employed'' by one of the authors in \cite{Fau19} for Conjecture \ref{conj:hexagonal}, namely {\em 'what else could it be?'}.

\begin{conjecture}\label{conj:Leech}
	We conjecture that a Leech lattice solves the quantum paving problem in phase space of dimension 24 for arbitrary fixed density greater than 1.
\end{conjecture}
Some arguments in favor of Conjecture \ref{conj:Leech} are the optimality of the Leech lattice for sphere packing \cite{Coh-Via17} and its universal optimality for energy minimization \cite{Coh-Via22} combined with the fact that it is still the best known solution to the sphere covering problem and locally optimal among lattice covering \cite{SchVal05}. An issue, however, which needs to be considered is that quantum paving is highly sensitive to permutations of the base vectors and rotations of the lattice (expect for phase space dimension 2). We refer, e.g., to \cite{PfaRas13} where some aspects of this problem are discussed. Hence, it is important to note that we conjecture that there is \textit{a} (not \textit{the}) Leech lattice which solves quantum paving, leaving the option that a permutation of the base vectors or a rotation of the classical Leech lattice is needed and that other versions may lead to a non-optimal solution.

To better understand the difficulty, consider the following example. Take $\Z^4$ with its canonical base in $\R^4$. Let $\alpha, \beta > 0$ and set, only for the moment,
	\begin{equation}
		\L_1(\alpha) = \alpha
		\begin{pmatrix}
			\beta & 0 & 0 & 0\\
			0 & \beta & 0 & 0\\
			0 & 0 & \frac{1}{\beta} & 0\\
			0 & 0 & 0 & \frac{1}{\beta}
		\end{pmatrix}
		 \Z^4
		\quad \text{ and } \quad
		\L_2(\alpha) = \alpha
		\begin{pmatrix}
			\beta & 0 & 0 & 0\\
			0 & \frac{1}{\beta} & 0 & 0\\
			0 & 0 & \beta & 0\\
			0 & 0 & 0 & \frac{1}{\beta}
		\end{pmatrix} \Z^4.
	\end{equation}
	As lattices, $\L_1(\alpha)$ and $\L_2(\alpha)$ would usually be identified. One is obtained from the other by re-labeling the co-ordinates (so they are equivalent by a permutation matrix). However, the matrix generating $\L_1(\alpha)$ is symplectic, whereas $\L_2(\alpha)$ is not a symplectic lattice (unless $\beta = 1$). A consequence for Gabor systems in $\Lt[2]$ with 2-dimensional Gaussian function $\varphi(t_1,t_2) = 2^{1/2} e^{-\pi(t_1^2+t_2^2)}$ is
	\begin{equation}
		\G(\varphi, \L_1(\alpha)) \text{ is a frame}
		\Longleftrightarrow
		\alpha < 1,
		\quad \text{ whereas } \quad
		\G(\varphi, \L_2(\alpha)) \text{ is a frame}
		\Longleftrightarrow
		\alpha \beta < 1 \, \wedge \, \alpha/\beta < 1.
	\end{equation}
	This follows from the characterization results on univariate Gaussian Gabor frames \cite{Lyu92}, \cite{Sei92}, \cite{SeiWal92} in combination with results on tensor Gabor systems \cite{Bou08} (see also \cite[Prop.~3]{PfaRas13}).

In phase space dimension 8 it is even more delicate to come up with a first guess. While the $\mathsf{E}_8$ lattice yields the best sphere packing in dimension 8 \cite{Via17} and is universally optimal for energy minimization \cite{Coh-Via22}, it provably fails to give the optimal sphere covering solution: it has worse covering density than the $\mathsf{A}_8^*$ lattice \cite{Conway} and is not even locally optimal \cite{SchVal05}. We pose the following question.
\question
	Does an $\mathsf{E}_8$ lattice solve the quantum paving problem in phase space of dimension 8 for arbitrary fixed density greater than 1?
Indeed, we expect the answer to be `\textit{no}' due to results and an argument (eq.~(1.1)) given in \cite{FauSha23}, and the non-optimality of $\mathsf{E}_8$ regarding the sphere covering problem.

\question
Is it possible that in some dimension we may have different solutions for different lattice densities for the quantum packing, covering, or paving problem? 

\remark The results in~\cite[Section 4]{Noh} give rise to speculation that quantum paving may also be of relevance in connection with quantum communication over Gaussian thermal loss channels, when combined with Gottesman-Kitaev-Preskill codes. We leave this admittedly highly speculative idea for future research.

\section*{Acknowledgement}

M.~Faulhuber was supported by the Austrian Science Fund (FWF) 10.55776/P33217. T.~Strohmer was supported by NSF DMS-2208356. We would like to thank Stefan Steinerberger for stimulating discussions  on the topic of this article.


\end{document}